\documentclass[creativecommons]{eptcs}

\usepackage{tikz}
\usetikzlibrary{positioning, arrows.meta, calc, fit, backgrounds}
\usepackage{graphicx}
\usepackage{booktabs}
\usepackage{stmaryrd}       
\usepackage{xspace}
\usepackage{amsmath,amssymb,amsthm}
\usepackage{mathpartir}
\usepackage{hyperref}
\usepackage[capitalise,noabbrev]{cleveref}

\providecommand{\doi}[1]{doi:\href{https://doi.org/#1}{\begingroup\urlstyle{rm}\nolinkurl{#1}\endgroup}}

\newcommand{\spar}{\mathbin{\|}}
\newcommand{\srec}[2]{\mu #1.\, #2}
\newcommand{\sbranch}[1]{\&\{#1\}}
\newcommand{\sselect}[1]{\oplus\{#1\}}
\newcommand{\statespace}[1]{\mathcal{L}(#1)}
\newcommand{\End}{\mathbf{end}}
\newcommand{\reach}{\twoheadrightarrow}
\newcommand{\reticulate}{\textsc{Reticulate}\xspace}
\newcommand{\camrev}[1]{#1}
\newcommand{\mungo}{\textsc{Mungo}\xspace}
\newcommand{\stmungo}{\textsc{StMungo}\xspace}

\tikzset{
  state/.style={circle, draw, minimum size=6mm, inner sep=1pt, font=\small},
  state top/.style={state, fill=blue!15},
  state bot/.style={state, fill=red!15},
  trans/.style={-{Stealth[length=5pt]}, font=\scriptsize},
  scc/.style={rounded corners=4pt, draw=gray, dashed, inner sep=4pt},
  sub/.style={rounded corners=4pt, draw=gray!60, fill=gray!5, inner sep=6pt},
  dot/.style={circle, draw, minimum size=3.4mm, inner sep=0pt},
}

\newtheoremstyle{spacedplain}%
  {1em plus 0.3em minus 0.2em}
  {0.7em plus 0.2em}
  {\itshape}
  {}
  {\bfseries}
  {.}
  {.5em}
  {\thmname{#1}\thmnumber{ #2}\thmnote{ (#3)}}
\theoremstyle{spacedplain}
\newtheorem{theorem}{Theorem}
\newtheorem{lemma}{Lemma}
\newtheorem{proposition}{Proposition}
\newtheorem{corollary}{Corollary}

\newtheoremstyle{spaceddef}%
  {1em plus 0.3em minus 0.2em}
  {0.7em plus 0.2em}
  {}
  {}
  {\bfseries}
  {.}
  {.5em}
  {\thmname{#1}\thmnumber{ #2}\thmnote{ (#3)}}
\theoremstyle{spaceddef}
\newtheorem{definition}{Definition}

\theoremstyle{remark}

\newtheoremstyle{conceptstyle}%
  {1em plus 0.3em minus 0.2em}
  {0.7em plus 0.2em}
  {}
  {}
  {\bfseries}
  {.}
  {.5em}
  {\thmnote{#3}}
\theoremstyle{conceptstyle}
\newtheorem*{concept}{}

\makeatletter
\let\c@lemma\c@theorem            
\let\c@proposition\c@theorem      
\let\c@corollary\c@theorem        
\let\c@definition\c@theorem       
\let\c@example\c@theorem          
\let\c@remark\c@theorem           
\makeatother

\crefname{theorem}{Theorem}{Theorems}
\Crefname{theorem}{Theorem}{Theorems}
\crefname{lemma}{Lemma}{Lemmas}
\Crefname{lemma}{Lemma}{Lemmas}
\crefname{proposition}{Proposition}{Propositions}
\Crefname{proposition}{Proposition}{Propositions}
\crefname{corollary}{Corollary}{Corollaries}
\Crefname{corollary}{Corollary}{Corollaries}
\crefname{definition}{Definition}{Definitions}
\Crefname{definition}{Definition}{Definitions}
\crefname{example}{Example}{Examples}
\Crefname{example}{Example}{Examples}
\crefname{remark}{Remark}{Remarks}
\Crefname{remark}{Remark}{Remarks}

\title{Session Type State Spaces Form Lattices}
\author{Alexandre Zua Caldeira
\institute{Independent Researcher, Berlin, Germany}
\email{zua@bica-tools.org}}

\def\titlerunning{Session Type State Spaces Form Lattices}
\def\authorrunning{A. Zua Caldeira}
\def\event{ICE 2026}

\begin{document}

\maketitle

\begin{abstract}
We prove that the state space of every well-formed session type,
quotiented by strongly connected components, forms a bounded
lattice; $n$-ary parallel composition yields product lattices.
Two consequences follow: duality preserves the lattice up to
isomorphism, and Gay--Hole width subtyping corresponds to lattice
embedding for non-recursive types.
We validate this on 108~benchmark protocols across networking,
databases, distributed systems, AI, and fault tolerance: all
form lattices, 93 distributive, 15 non-distributive. Mechanised
in Lean~4 with two independently developed tool implementations.
\end{abstract}

\section{Introduction}\label{sec:intro}

Session types~\cite{honda1998,vasconcelos2012sessions} describe
the communication protocols that govern interaction with objects,
channels, or services.
A session type specifies the allowed sequences of
operations\,---\,method calls, branches, selections\,---\,together with
their ordering constraints.
The \emph{state space} of a session type is a labelled transition
system (LTS) \emph{over states} obtained by ``executing'' the
type: each state represents a protocol stage, and each transition
a permitted action. The semantics is specification-level rather than
operational: states are the unit of reachability, not reduction
on terms.

State spaces are routinely used for verification, subtype checking,
and protocol analysis, yet their algebraic structure has received
little attention.
Gay and Hole~\cite{gay2005subtyping}, for instance, study subtyping
as a preorder on session types but do not investigate the structure
of the underlying state spaces.
In this paper we show that \emph{the state space of every
well-formed session type, after quotienting by strongly connected
components, is a bounded lattice}.
We call this lattice a \emph{reticulate}\,---\,from Latin
\emph{reticulatus}, the root of ``lattice'', giving session
type lattices their own name while remaining etymologically
precise.

The lattice structure arises naturally from the session type
constructors: branch (\&) and selection ($\oplus$) create shared
extrema that serve as meets and joins; recursion introduces cycles
absorbed by the
\emph{strongly-connected-component (SCC) quotient}\,---\,the standard
graph-theoretic operation that collapses each set of mutually
reachable states into one equivalence class, yielding a directed
acyclic graph (DAG); and $n$-ary parallel composition ($\|$) yields
$n$-dimensional product lattices.
The $\|$ constructor is where lattice structure becomes
\emph{necessary} rather than merely convenient.
Without~$\|$, the SCC quotient of a session type's state space is
a tree (pure sequencing), a diamond (binary choice), or a DAG
(nested choice).
With~$\|$, the $n$-fold product of state spaces inherently has
meets and joins corresponding to fork and synchronisation points
of concurrent execution.

Our contributions are:
(1) the \emph{Reticulate Theorem}: the SCC quotient of every
well-formed session type's state space is a bounded lattice,
by structural induction (\cref{sec:statespace});
(2) \emph{duality preservation}: swapping branch and selection
preserves state-space isomorphism (\cref{sec:duality});
(3) \emph{subtyping as lattice embedding}: \camrev{Gay--Hole width
subtyping is a label-preserving simulation on the state space whose
SCC-quotient projection is a sound lattice embedding} for non-recursive
types, with the selection case derived from the branch case via
duality (\cref{sec:subtyping});
(4) \emph{validation} (\cref{sec:impl-eval}): an empirical
study on 108~benchmark protocols using two independently
developed tools; mechanised verification of the proof in
Lean~4 (17~modules, zero \texttt{sorry}); and a distributivity
classification: $N_5$ from depth-unequal choice, $M_3$ from
wide choice ($n \geq 3$), which together account for all 15
non-distributive cases.

\Cref{sec:background} recalls the relevant background.
\Cref{sec:statespace} defines the state-space construction and
presents the proof.
\Cref{sec:consequences} develops the duality and subtyping
applications.
\Cref{sec:casestudy} demonstrates each result on a concrete
concurrent file protocol.
\Cref{sec:impl-eval} reports the implementation, the Lean~4
mechanisation, and the empirical study on 108~protocols.
\Cref{sec:related} discusses related work and
\cref{sec:conclusion} concludes.

\section{Preliminaries}\label{sec:background}

\subsection{Session Types}\label{sec:syntax}

The grammar below targets method-call protocols on a single shared
\emph{object}; channel-based bi- and multiparty communication
(send/receive, delegation, as in process
calculi~\cite{honda1998}) is out of scope unless the channel is
itself modelled as an object, in which case the present framework
applies directly.
We restrict attention to finite-state (regular) session types.

The grammar below extends typestate session
types~\cite{bica2010,mungo2016} with
parallel composition, letting the lattice construction reach
objects shared by concurrent clients:
\begin{align*}
S \;\;::=\;\; & \sbranch{m_1 : S_1,\; \ldots,\; m_n : S_n}
                  & \text{(branch, $n \geq 0$)} \\
  \mid\;\; & \sselect{l_1 : S_1,\; \ldots,\; l_n : S_n}
                  & \text{(selection, $n \geq 0$)} \\
  \mid\;\; & S_1 \spar \cdots \spar S_n
                  & \text{(parallel, $n \geq 2$)} \\
  \mid\;\; & \srec{X}{S}
                  & \text{(recursion)} \\
  \mid\;\; & X
                  & \text{(variable)} \\
  \mid\;\; & \End
                  & \text{(termination)} \\[6pt]
D \;\;::=\;\; & S \;\mid\; S, \; X_1 = S_1, \; \ldots, \; X_n = S_n
                  & \text{(declaration, equations $n \geq 0$)}
\end{align*}
A declaration~$D$ is a root session type~$S$ optionally followed by
named equations $X_i = S_i$. The equations decompose a protocol
into named, reusable phases; cyclic references express recursion
(including mutual recursion), subsuming $\srec{X}{S}$ as the special
case $X = S$. We retain
$\srec{X}{S}$ as a separate primitive, equivalent to the
single-equation declaration $X = S$.

The branch $\sbranch{m_1:S_1, \ldots, m_n:S_n}$ represents a state
in which the object unconditionally offers $n$ method calls; the
client exercises an external choice of which to invoke, and the
protocol continues as $S_i$ for the chosen $m_i$ ($n = 0$ is the
empty offer). In contrast, the selection
$\sselect{l_1:S_1, \ldots, l_n:S_n}$ represents an internal choice
in which the object returns one of $n$ enumeration labels $l_i$;
the client must dispatch on the returned label and proceed with
the corresponding continuation $S_i$ ($n = 0$ is the empty
internal choice). Parallel
composition $S_1 \spar \cdots \spar S_n$ is the flat $n$-ary form
(no nesting), with each $S_i$ describing one client's interactions
with the shared object; the sub-protocols run concurrently and may
cooperate through that object, with clients playing dual or
complementary roles (e.g., reader/writer, publish/subscribe). A variable $X$ is a reference to its
binder, either a $\srec{X}{\cdot}$ form or an equation $X = S$.

The constant $\End$ is the explicit termination primitive, not
synonymous with the empty branch $\sbranch{}$: $\sbranch{}$ merely
denotes a state with no enabled methods. Every terminated session has no methods, but the
converse fails\,---\,a parallel branch that has exhausted its
actions has no methods to offer yet is not ``terminated'' until its
siblings complete. We retain $\End$ as a primitive to preserve this
asymmetry, matching the additive units of classical linear
logic~\cite{wadler2012propositions}: $\sbranch{} \cong \top$,
$\sselect{} \cong \mathbf{0}$, dual via $\mathbf{0}^\perp = \top$.

\begin{definition}[Well-formedness]\label{def:wellformed}
A session type (or declaration) is \emph{well-formed} if all of
the following hold.
\begin{enumerate}
  \item \textbf{Closedness.} $S$ (resp.\ $D$) has no free
    variables, where binders are $\srec{X}{\cdot}$ and
    equation left-hand sides~$X$.
  \item \textbf{Contractiveness} (also called
    \emph{guardedness} in the process-calculus tradition).
    Every recursion
    $\srec{X}{S}$ has at least one constructor between
    $\srec{X}{\cdot}$ and each occurrence of $X$ in $S$; in
    equation systems, every cyclic reference path passes
    through at least one constructor.
  \item \textbf{Termination.} From every state, there is a
    path to $\End$.
  \item \textbf{Parallel closedness.} In every parallel
    composition $S_1 \spar \cdots \spar S_n$ occurring in
    $S$, no variable bound outside the parallel occurs free
    in any~$S_i$.
\end{enumerate}
\end{definition}

Closedness, contractiveness, and parallel closedness are standard
in the session-type literature~\cite{honda1998,vasconcelos2012sessions,mungo2016,bica2010}
and rule out genuinely ill-formed types\,---\,dangling references,
uncontractive recursions like $\srec{X}{X}$, and parallel branches
that capture outer recursion variables (\cref{lem:parallel}).
Termination is specific to our setting: it is necessary for the
bounded-lattice property (\cref{thm:reticulate}), not for
admissibility. Non-terminating types remain typeable, with both
tools exposing identical \texttt{-{}-non-termination} handling
(\cref{sec:impl-eval}).

\subsection{Lattice Theory}\label{sec:lattice-theory}

We recall the standard lattice-theoretic vocabulary used in
the proofs of \cref{thm:reticulate};
see~\cite{birkhoff1967lattice} for a textbook treatment.

\begin{concept}[Bounded lattice]
A bounded lattice is a poset $(L, \leq)$ with a least element
$\bot$ and a greatest element $\top$ such that every pair of
elements has a \emph{meet} (greatest lower bound) $a \wedge b$
and a \emph{join} (least upper bound) $a \vee b$.
\end{concept}

\begin{concept}[Product lattice]
Given bounded lattices $(L_1, \leq_1), \ldots, (L_n, \leq_n)$,
their product $L_1 \times \cdots \times L_n$ is ordered
componentwise: $(a_1, \ldots, a_n) \leq (b_1, \ldots, b_n)$ iff
$a_i \leq_i b_i$ for all~$i$. Meets, joins, top, and bottom are
componentwise, and the product of finitely many bounded
lattices is itself a bounded lattice.
\end{concept}

\begin{concept}[Distributive lattice]
A bounded lattice is \emph{distributive} if
$a \wedge (b \vee c) = (a \wedge b) \vee (a \wedge c)$ holds
for all elements; equivalently, if it contains no sublattice
isomorphic to $N_5$ (the pentagon) or $M_3$ (the diamond
with three atoms).
\end{concept}

\begin{concept}[Upward-closed subset]
A subset $D$ of a poset $(L, \leq)$ is \emph{upward-closed} if
$x \in D$ and $x \leq y$ imply $y \in D$; it is \emph{proper}
when $D \neq L$. The dual notion is \emph{downward-closed}.
\end{concept}

\begin{concept}[SCC quotient]
A \emph{strongly connected component} (SCC) of a directed graph
is a maximal set of mutually reachable vertices; the \emph{SCC
quotient} collapses each SCC to a single node, yielding a
DAG~\cite{tarjan1972scc}. In our setting, SCCs collect the
states a recursive protocol can revisit; on the quotient,
reachability defines a partial order, lifted to a bounded
lattice in \cref{lem:recursion}.
\end{concept}

\section{State Space Construction and the Reticulate Theorem}\label{sec:statespace}

\camrev{The construction is a two-step pipeline:}
\[
  \underset{\camrev{\text{type}}}{S}
  \;\xrightarrow{\camrev{\text{\,construct\,}}}\;
  \underset{\camrev{\substack{\text{state space}\\\text{(LTS, cyclic)}}}}{\statespace{S}}
  \;\xrightarrow{\camrev{\substack{\text{order by reachability,}\\\text{then SCC quotient}}}}\;
  \underset{\camrev{\substack{\text{bounded lattice}\\\text{the \emph{reticulate}}}}}{\statespace{S}/{\equiv}}
\]
\camrev{The two steps are the first two subsections: the state space is
built by structural induction (\cref{sec:construction}); then, after its
states are ordered by reachability, the SCC quotient collapses it into the
bounded lattice (\cref{sec:scaffolding}). \Cref{sec:reticulate-theorem}
then composes the per-constructor lemmas into the headline theorem
(\cref{thm:reticulate}). Proof sketches appear inline; full proofs are in
Appendix~\ref{appendix:proofs}.}

\subsection{State spaces and the inductive construction}\label{sec:construction}

\begin{definition}[State space]\label{def:statespace}
A state space is a tuple
$(Q, \Sigma, \delta, q_\top, q_\bot)$ where $Q$
is a finite set of states, $\Sigma$ is a finite alphabet of action
labels, $\delta \subseteq Q \times \Sigma \times Q$ is the
transition relation, $q_\top \in Q$ is the initial state, and
$q_\bot \in Q$ is the terminal state. The construction of
\cref{def:construction} assigns a state space $\mathcal{L}(S)$ to
each well-formed session type~$S$.
\end{definition}

\begin{definition}[State-space construction $\mathcal{L}(S)$]%
\label{def:construction}
The state space $\mathcal{L}(S) = (Q, \Sigma, \delta,
q_\top, q_\bot)$ is defined by structural induction on
closed well-formed~$S$ (\cref{def:wellformed}).
\camrev{The induction yields state spaces, not bounded lattices.
Of the grammar's two base cases, $\End$ and~$X$, only $\End$ takes
a clause below: under closedness every bare~$X$ is bound (by an
enclosing $\srec{X}{\cdot}$ or a definition $X = S'$), so each
occurrence is built as a \emph{placeholder} state and then
eliminated by cases~(4)--(5), which redirect it as a back-edge to
the already-built entry state.}
For closed well-formed $S$, the SCC quotient
$\mathcal{L}(S)/{\equiv}$ is a bounded lattice
(\cref{def:scc-quotient,thm:reticulate}).

\begin{enumerate}
  \item \textbf{End.}
    $\mathcal{L}(\End) = (\{q_0\}, \emptyset, \emptyset, q_0, q_0)$
    \,---\, a single state, both initial and terminal.

  \item \textbf{Branch} $\sbranch{m_i : S_i}_{i=1..n}$.
    \begin{align*}
      Q        &= \{q_0\} \cup \textstyle\bigcup_i
                  (Q_i \setminus \{q_\bot^i\}) \cup \{q_\bot\}, \\
      \Sigma   &= \{m_1,\ldots,m_n\} \cup \textstyle\bigcup_i \Sigma_i, \\
      \delta   &= \{(q_0, m_i, q_\top^i)\}_i
                  \cup \textstyle\bigcup_i \delta_i[q_\bot^i := q_\bot],
    \end{align*}
    where $Q_i, \Sigma_i, \delta_i, q_\top^i, q_\bot^i$ are the
    components of $\mathcal{L}(S_i)$; the substitution
    $[q_\bot^i := q_\bot]$ shares $q_\bot$ across children;
    $q_\top = q_0$. For $n = 0$, $Q = \{q_0, q_\bot\}$,
    $\Sigma = \{\tau\}$, $\delta = \{(q_0, \tau, q_\bot)\}$ with
    $\tau$ a silent action for vacuous termination; the SCC
    quotient is the chain $q_0 > q_\bot$.

  \item \textbf{Selection} $\sselect{l_i : S_i}_{i=1..n}$.
    Identical tuple to branch at the state-space level: the
    construction depends only on the continuations
    (\cref{lem:polarity}), not on polarity.

  \item \textbf{Recursion} $\srec{X}{S'}$.
    Build $\mathcal{L}(S') = (Q', \Sigma', \delta', q_\top',
    q_\bot')$ with fresh placeholder states for each occurrence
    of~$X$; let $P_X \subseteq Q'$ collect these placeholders.
    Then
    $Q = Q' \setminus P_X$,
    $\Sigma = \Sigma'$,
    $\delta = \delta'[P_X := q_\top']$ (redirect every transition
    targeting a placeholder to $q_\top'$, creating back-edges),
    $q_\top = q_\top'$, $q_\bot = q_\bot'$.

  \item \textbf{Named definitions}
    $D = S, X_1 = S_1, \ldots, X_k = S_k$ (root session type
    followed by $k \geq 0$ equations).
    Build $\mathcal{L}(S)$ and each $\mathcal{L}(S_j)$ with
    placeholders $p_{i,j}$ for each occurrence of $X_j$. For each
    name $X_j$, replace every transition targeting any $p_{*,j}$
    with a transition to $q_\top^{S_j}$, and remove the
    placeholders. The result is
    $\mathcal{L}(D) = (Q, \Sigma, \delta, q_\top^S, q_\bot^S)$
    where $Q = (Q^S \cup \bigcup_j Q^{S_j})$ with placeholders
    removed, $\Sigma = \Sigma^S \cup \bigcup_j \Sigma^{S_j}$,
    and $\delta$ is the global redirection.

  \item \textbf{Parallel} $S_1 \spar \cdots \spar S_n$.%
    \label{def:product-construction}
    $\mathcal{L}(S_1 \spar \cdots \spar S_n) =
      (Q_1 \times \cdots \times Q_n, \Sigma, \delta_\times,
       (q_\top^1,\ldots,q_\top^n), (q_\bot^1,\ldots,q_\bot^n))$
    with $\Sigma = \bigcup_i \Sigma_i$ and $\delta_\times$ the
    componentwise advance relation:
    $(s_1, \ldots, s_i, \ldots, s_n) \xrightarrow{a}
    (s_1, \ldots, s_i', \ldots, s_n)$ for each
    $s_i \xrightarrow{a} s_i'$ in $\delta_i$.
\end{enumerate}
\end{definition}

\paragraph{Inductive, not coinductive.}
\camrev{\Cref{def:construction} is structural recursion on the finite
six-constructor syntax; the carrier is state spaces, not bounded lattices.
Recursion uses a back-edge rather than unfolding: for $\srec{X}{S}$ the body
$S$ is built once and each occurrence of $X$ is redirected to the top
$q_\top$ (cases~4--5), forming a finite cycle. The construction never forms
the infinite unfolding\,---\,the greatest fixed point\,---\,and uses no
coinduction; the SCC quotient collapses the cycle. The Lean mechanisation is
likewise coinduction-free.}

\subsection{SCC quotienting and the bounded lattice}\label{sec:scaffolding}

\camrev{A bounded lattice is a poset with a least and a greatest element and a
meet and a join for every pair (\cref{sec:lattice-theory}). We establish
these for $\statespace{S}/{\equiv}$: reachability orders it as a poset
(\cref{def:reach,def:scc-quotient}); \cref{prop:extrema} gives its least and
greatest elements; and the meets and joins are supplied constructor by
constructor, by structural induction on $S$, assembled by the Reticulate
Theorem (\cref{thm:reticulate}).}

\begin{definition}[Reachability]\label{def:reach}
Write $s_1 \reach s_2$ (``$s_1$ reaches $s_2$'') iff there
exists a path from $s_1$ to~$s_2$ in $(Q, \delta)$.
\end{definition}

\begin{proposition}[Preorder]\label{prop:reach-preorder}
The relation $\reach$ is reflexive and transitive, hence a
preorder; cycles from recursion prevent antisymmetry.
\end{proposition}

\begin{definition}[SCC quotient and its order]\label{def:scc-quotient}
\camrev{The \emph{SCC quotient} $\statespace{S}/{\equiv}$ identifies mutually
reachable states, collapsing each recursion cycle to one class and yielding
an antisymmetric quotient. It is ordered by} $[s_1] \leq [s_2]$ iff
$s_2 \reach s_1$: a state is smaller when less of the protocol
remains, so execution descends from $q_\top$ to $q_\bot$.
\end{definition}

\begin{proposition}[Extrema]\label{prop:extrema}
\camrev{For every well-formed $S$, the initial state $q_\top$ reaches every
state and the terminal state $q_\bot$ is reachable from every state; hence
the quotient $\statespace{S}/{\equiv}$ has greatest element $q_\top$ and
least element $q_\bot$.}
\end{proposition}

\camrev{It remains to exhibit a meet and a join for every pair.}
Each grammar constructor corresponds to a lattice operation on its
sub-lattices: $\End$
yields the one-element lattice; choice glues sub-lattices through a
fresh top into a diamond; recursion absorbs an upward-closed set into
its body's top via the SCC quotient; named definitions iterate that
absorption; parallel takes a Cartesian product. Three structural
innovations make this work jointly: \emph{end-identification}
(singleton $\statespace{\End}$ with shared $q_\bot$ across branch and
selection children; \cref{lem:end}, \cref{lem:choice}), the
\emph{SCC quotient} (\cref{def:scc-quotient}; \cref{lem:recursion},
\cref{lem:equations}), and the \emph{$\|$-product} (\cref{lem:parallel}).
Each is load-bearing on at least one constructor case.
\camrev{Each of the
following lemmas treats one constructor, assuming the sub-state-spaces are
bounded lattices and extending the property to the composite;
\cref{thm:reticulate} composes the cases.}

\begin{lemma}[End]\label{lem:end}
$\statespace{\End}/{\equiv}$ is a bounded lattice.
\end{lemma}

\begin{lemma}[Polarity erasure]\label{lem:polarity}
For any labels $a_1, \ldots, a_n$ and well-formed types
$S_1, \ldots, S_n$,
\[
\statespace{\sbranch{a_1{:}S_1, \ldots, a_n{:}S_n}} \;=\;
\statespace{\sselect{a_1{:}S_1, \ldots, a_n{:}S_n}}.
\]
\end{lemma}

\begin{proof}
Both constructions create a fresh $q_\top$ with $n$ transitions
to sub-state-space entries, sharing $q_\bot$ — depending only on
the continuations, not on the polarity. Polarity is a typing
discipline, not a structural property of the state space.
\end{proof}

\begin{lemma}[Choice]\label{lem:choice}
If $\statespace{S_i}/{\equiv}$ is a bounded lattice for each~$i$, then
\camrev{so is $\statespace{\sbranch{a_1 : S_1, \ldots, a_n : S_n}}/{\equiv}$;
by \cref{lem:polarity} the same holds for
$\sselect{a_1 : S_1, \ldots, a_n : S_n}$}.
\end{lemma}

\begin{proof}
\camrev{Adjoining a common top $q_\top$ and bottom $q_\bot$ turns the
sub-lattices into a single bounded lattice: meets and joins within a
sub-lattice are inherited, and across sub-lattices, which are disjoint
above $q_\bot$, $q_\top$ is the only common upper bound and $q_\bot$ the
only common lower bound (\cref{fig:diamond}).}
\end{proof}

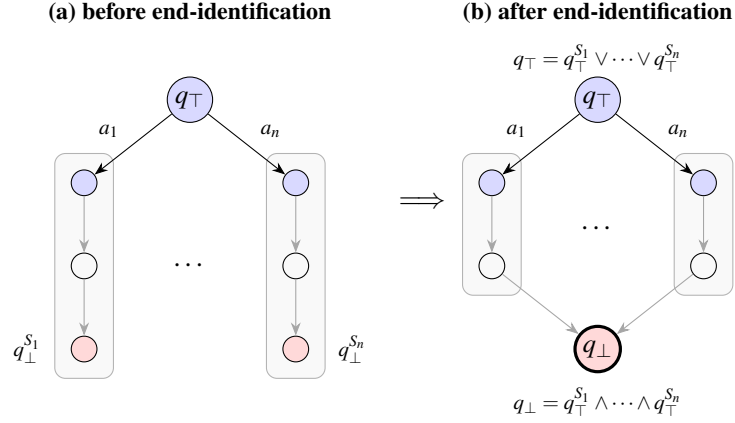
\begin{figure}[t]
\centering
\begin{tikzpicture}[node distance=12mm, baseline=(top.base)]
  \node[font=\footnotesize\bfseries] at (0,3.55) {(a) before end-identification};
  \node[state top] (top) at (0,2.4) {$q_\top$};
  \foreach \i/\x in {1/-1.4, n/1.4} {
     \node[dot, fill=blue!15] (t\i) at (\x,1.3) {};
     \node[dot] (m\i) at (\x,0.2) {};
     \node[dot, fill=red!15] (b\i) at (\x,-0.9) {};
     \draw[trans, gray!70] (t\i) -- (m\i);
     \draw[trans, gray!70] (m\i) -- (b\i);
     \begin{scope}[on background layer]
       \node[sub, fit=(t\i)(m\i)(b\i)] {};
     \end{scope}
  }
  \node at (0,0.2) {$\cdots$};
  \node[font=\scriptsize] at (-2.15,-0.9) {$q^{S_1}_\bot$};
  \node[font=\scriptsize] at (2.15,-0.9) {$q^{S_n}_\bot$};
  \draw[trans] (top) -- node[above left,pos=0.55] {$a_1$} (t1);
  \draw[trans] (top) -- node[above right,pos=0.55] {$a_n$} (tn);
\end{tikzpicture}%
\hspace{3mm}\raisebox{-1.4cm}{$\Longrightarrow$}\hspace{1mm}%
\begin{tikzpicture}[node distance=12mm, baseline=(top.base)]
  \node[font=\footnotesize\bfseries] at (0,3.55) {(b) after end-identification};
  \node[font=\scriptsize] at (0,2.95) {$q_\top = q^{S_1}_\top \vee \cdots \vee q^{S_n}_\top$};
  \node[state top] (top) at (0,2.4) {$q_\top$};
  \foreach \i/\x in {1/-1.4, n/1.4} {
     \node[dot, fill=blue!15] (t\i) at (\x,1.3) {};
     \node[dot] (m\i) at (\x,0.2) {};
     \draw[trans, gray!70] (t\i) -- (m\i);
     \begin{scope}[on background layer]
       \node[sub, fit=(t\i)(m\i)] {};
     \end{scope}
  }
  \node at (0,0.7) {$\cdots$};
  \node[state bot, very thick] (bot) at (0,-0.9) {$q_\bot$};
  \node[font=\scriptsize] at (0,-1.6) {$q_\bot = q^{S_1}_\top \wedge \cdots \wedge q^{S_n}_\top$};
  \draw[trans, gray!70] (m1) -- (bot);
  \draw[trans, gray!70] (mn) -- (bot);
  \draw[trans] (top) -- node[above left,pos=0.55] {$a_1$} (t1);
  \draw[trans] (top) -- node[above right,pos=0.55] {$a_n$} (tn);
\end{tikzpicture}
\caption{\camrev{%
  Constructing $\statespace{\sbranch{a_1{:}S_1,\dots,a_n{:}S_n}}$.
  (a)~A fresh $q_\top$ fans by one transition $a_i$ into each
  sub-lattice $\statespace{S_i}$, each with its own least element
  $q^{S_i}_\bot$.
  (b)~End-identification glues every $q^{S_i}_\bot$ onto the single
  reserved terminal $q_\bot$; the fresh $q_\top$ is then the only
  common upper bound and $q_\bot$ the only common lower bound,
  making the disjoint sub-lattices one bounded lattice.
  For $n = 2$ with singleton arms this is the diamond lattice $M_2$.
  The join and meet annotations in~(b) assume $n \geq 2$; at $n = 1$ the
  fresh $q_\top$ simply prepends one element above $\statespace{S_1}$.}}
\label{fig:diamond}
\end{figure}

\camrev{Three lemmas cover recursion and named equations.
\Cref{lem:absorption} is the order-theoretic tool: collapsing a
proper upward-closed set onto $q_\top$ preserves boundedness.
\Cref{lem:recursion} applies it to a single recursive binder,
whose back-edges form an SCC absorbed into $[q_\top]$ (\cref{fig:smtp}).
\Cref{lem:equations} treats named declarations, which are more
expressive than $\srec{X}{S}$: cyclic references (single or mutual
recursion) reuse the same absorption, while an acyclic shared name
--- like $\End$ --- collapses duplicated subtrees, giving lattice
rather than tree structure.}

\begin{lemma}[Top absorption]\label{lem:absorption}
Let $(L, \leq)$ be a bounded lattice and $D$ a \emph{proper}
upward-closed subset of $L$ containing $\top$
(equivalently, $\bot \notin D$).
Define $L' = (L \setminus D) \cup \{\top\}$ with the restricted
ordering.
Then $L'$ is a bounded lattice.
\end{lemma}

\begin{proof}
Since $D$ is upward-closed and proper, $\bot \notin D$, so
$L'$ retains both extrema. Meets in $L \setminus D$ stay in
$L \setminus D$ (the complement of an upward-closed set is
downward-closed) and are preserved. Joins in $L \setminus D$
are likewise preserved, except those whose $L$-join lies in
$D$, which collapse to $\top$ in $L'$ — any non-$\top$ upper
bound would itself lie in $D$ by upward-closure.
\end{proof}

\begin{lemma}[Recursion]\label{lem:recursion}
Let $S$ be a body with a single free variable $X$, such that the
state space $\mathcal{L}(S)$ with placeholder state $p_X$
satisfies: (i)~$q_\top^S$ reaches every non-placeholder state;
(ii)~the placeholder-free SCC quotient is a bounded lattice; and
(iii)~the recursion $\srec{X}{S}$ satisfies the
\textbf{Termination} clause of \cref{def:wellformed}.
Then $\statespace{\srec{X}{S}}/{\equiv}$ is a bounded lattice.
\end{lemma}

\begin{proof}
Adding back-edges from each placeholder $p_X$ to $q_\top^S$
collapses these states with $q_\top^S$ into one SCC class.
The set $D$ of quotient classes absorbed into $[q_\top^S]$
is upward-closed (mutual reachability through $q_\top^S$
propagates upward), and the Termination clause
(\cref{def:wellformed}) ensures $[q_\bot] \notin D$, making
$D$ a \emph{proper} upward-closed subset of the body quotient.
\Cref{lem:absorption} then yields the lattice.
\end{proof}

\begin{lemma}[Named definitions]\label{lem:equations}
Let $D = S, X_1 = S_1, \ldots, X_k = S_k$ be a well-formed
declaration \camrev{in which every equation is used either for
\emph{recursion} (each $X_j$ lies on a cyclic reference path) or
\emph{at most once}}. If $\statespace{S}/{\equiv}$ and each
$\statespace{S_i}/{\equiv}$ are bounded lattices, then
$\statespace{D}/{\equiv}$ is a bounded lattice. \camrev{More
generally it suffices that the sharing be \emph{non-reconvergent}:
no non-recursive name is reached from two distinct branches of the
same choice.}
\end{lemma}

\begin{proof}
\camrev{Three cases, by how each name is used.
A \emph{cyclic} reference behaves like recursion: its back-edges
merge states into an SCC, which \cref{lem:absorption} collapses
(one absorption per SCC, in any order), exactly as in
\cref{lem:recursion}.
A \emph{non-recursive name used along a single path} inlines its
sub-lattice beneath the one state that references it, adding no new
incomparable pairs, so every meet and join survives.
The remaining case is \emph{reconvergence}: when one name is
reached from two distinct branches of a choice, those branches
rejoin at the name's entry, leaving two incomparable states with
two incomparable common lower bounds, so their meet does not exist. Such declarations are
only bounded \emph{partial} lattices, handled by the Reticulate
Theorem (\cref{sec:reticulate-theorem}).}
\end{proof}

\begin{figure}[t]
\centering
\begin{tikzpicture}[node distance=12mm]
  \node[font=\footnotesize\bfseries] at (0,0.8) {(a) Before quotient};
  \node[state top] (q0) at (0,0) {$q_0$};
  \node[state] (q1) at (-1.5,-1.5) {$q_1$};
  \node[state bot] (qb) at (1.5,-1.5) {$q_\bot$};

  \draw[trans] (q0) -- node[above left] {\textit{mail}} (q1);
  \draw[trans] (q0) -- node[above right] {\textit{quit}} (qb);
  \draw[trans, dashed, gray] (q1) to[bend left=40]
    node[left, font=\scriptsize, text=gray] {\textit{send}} (q0);

  \begin{scope}[on background layer]
    \node[scc, fit=(q0)(q1),
      label={[font=\scriptsize,text=gray]left:SCC}] {};
  \end{scope}
\end{tikzpicture}
\hspace{16mm}
\begin{tikzpicture}[node distance=12mm]
  \node[font=\footnotesize\bfseries] at (0,0.8) {(b) After SCC quotient};
  \node[state top] (scc) at (0,0) {$[q_0,q_1]$};
  \node[state bot] (qb) at (0,-1.8) {$q_\bot$};

  \draw[trans] (scc) -- node[right] {\textit{quit}} (qb);
\end{tikzpicture}
\caption{%
  SMTP session type $\srec{X}{\sbranch{\mathit{mail} :
  \sbranch{\mathit{send} : X},\; \mathit{quit} : \End}}$.
  (a)~Before quotient: the recursion redirects \textit{send}
  from the placeholder to $q_0$, creating a cycle
  $q_0 \to q_1 \to q_0$ (dashed).
  (b)~After SCC quotient: the cycle collapses into $[q_0, q_1]$,
  yielding a 2-element chain. \camrev{Here $D = \{q_0, q_1\}$ is
  upward-closed and omits $q_\bot$, so \cref{lem:absorption}
  applies; the tool confirms 3 states, 3 transitions.}}
\label{fig:smtp}
\end{figure}
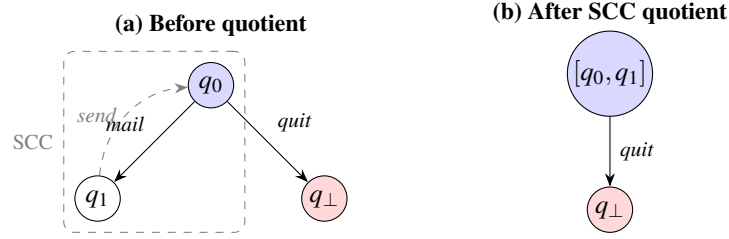

\begin{lemma}[Parallel]\label{lem:parallel}
If $\statespace{S_i}/{\equiv}$ is a bounded lattice for each
$i \in \{1, \ldots, n\}$, then
$\statespace{S_1 \spar \cdots \spar S_n}/{\equiv}$ is a bounded
lattice.
\end{lemma}

\begin{proof}
By the parallel case of \cref{def:construction} and the
independence of the $n$ branches (Termination ensures every
branch reaches $\End$; Parallel closedness ensures no shared
variables), the product inherits SCC structure componentwise:
$\statespace{S_1 \spar \cdots \spar S_n}/{\equiv} \cong
\prod_{i} \statespace{S_i}/{\equiv_i}$.
The product of finitely many bounded lattices is itself a
bounded lattice (\cref{sec:lattice-theory}), illustrated for
$n = 2$ in \cref{fig:product}.
\end{proof}

\begin{figure}[!hbt]
\centering
\begin{tikzpicture}[node distance=8mm, every node/.style={font=\small}]
  \node[state top] (00) at (0,0) {\scriptsize 0,0};
  \node[state] (10) at (-1.0,-0.9) {\scriptsize 1,0};
  \node[state] (01) at ( 1.0,-0.9) {\scriptsize 0,1};
  \node[state] (20) at (-2.0,-1.8) {\scriptsize 2,0};
  \node[state] (11) at ( 0,-1.8)   {\scriptsize 1,1};
  \node[state] (02) at ( 2.0,-1.8) {\scriptsize 0,2};
  \node[state] (21) at (-1.0,-2.7) {\scriptsize 2,1};
  \node[state] (12) at ( 1.0,-2.7) {\scriptsize 1,2};
  \node[state bot] (22) at (0,-3.6) {\scriptsize 2,2};

  \draw[trans] (00) -- node[above left] {$a$} (10);
  \draw[trans] (00) -- node[above right] {$c$} (01);
  \draw[trans] (10) -- node[above left] {$b$} (20);
  \draw[trans] (10) -- node[right] {$c$} (11);
  \draw[trans] (01) -- node[left] {$a$} (11);
  \draw[trans] (01) -- node[above right] {$d$} (02);
  \draw[trans] (20) -- node[right] {$c$} (21);
  \draw[trans] (11) -- node[above left] {$b$} (21);
  \draw[trans] (11) -- node[above right] {$d$} (12);
  \draw[trans] (02) -- node[left] {$a$} (12);
  \draw[trans] (21) -- node[below left] {$d$} (22);
  \draw[trans] (12) -- node[below right] {$b$} (22);
\end{tikzpicture}
\caption{%
  Product lattice
  $\statespace{\sbranch{a : \sbranch{b : \End}} \spar
  \sbranch{c : \sbranch{d : \End}}}$.
  The 9 states form a $3 \times 3$ grid lattice with
  componentwise ordering.}
\label{fig:product}
\end{figure}
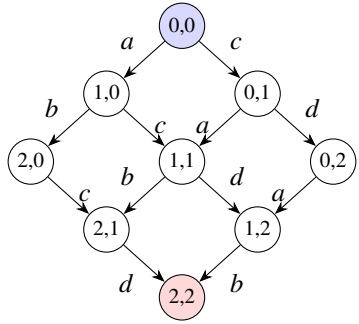

\subsection{The Reticulate Theorem}\label{sec:reticulate-theorem}

The per-constructor lemmas now compose into the headline result,
followed by tightness and a remark on realisability.

\begin{theorem}[Reticulate Theorem]\label{thm:reticulate}
For every well-formed session type~$S$,
$\statespace{S}/{\equiv}$ is a bounded lattice.
\end{theorem}

\begin{proof}
By structural induction on $S$, applying the per-constructor
\cref{lem:end,lem:choice,lem:recursion,lem:parallel}.
Each case preserves the bounded-lattice property, completing
the induction.
\end{proof}

\camrev{The theorem is stated for session types; for declarations with
\emph{named equations} the construction yields a bounded \emph{partial}
lattice in general. When two equations are each referenced from two distinct
choice branches, the branches reconverge and a meet or join may fail to
exist; recursive sharing is
absorbed by the SCC quotient and stays a lattice. Characterising the sharing
patterns that keep the operations total is open (\cref{sec:conclusion}).}

\begin{corollary}[Termination tightness]\label{cor:wfpar}
The \textbf{Termination} clause is tight: relaxing the
exit-path condition yields types whose SCC-quotiented state
spaces are not lattices.
\end{corollary}

A witness:
$\sbranch{a : \End,\; b : \srec{X}{\sbranch{a : X}}}$
violates \textbf{Termination}\,---\,from the body root
$\sbranch{a : X}$, the only outgoing transition unfolds $X$
back to the same root, so no path to $\End$ exists. After SCC
quotient the recursive branch collapses to a single class with
no outgoing transition to a terminal, leaving no bottom element
reachable from all states.

\begin{proposition}[A realisable fragment]%
\label{prop:realisability}
\camrev{Every finite chain is a reticulate, and reticulates are closed
under finite products,
$\statespace{S_1 \spar \cdots \spar S_n}/{\equiv} \cong
\prod_i \statespace{S_i}/{\equiv}$. Hence every finite product of chains
is a reticulate.}
\end{proposition}

\begin{proof}
\camrev{A single-arm branch prepends a fresh top above its continuation
(\cref{lem:choice} at $n = 1$), so $k$ nested single-arm branches over
$\statespace{\End}$ realise the $(k{+}1)$-element chain; closure under
finite products is \cref{lem:parallel}.}
\end{proof}

\section{Duality and Subtyping as Lattice Consequences}\label{sec:consequences}

The lattice structure has two consequences for class session types:
duality is invisible at the lattice level (\cref{sec:duality}); width
subtyping reduces to lattice embedding (\cref{sec:subtyping}).
On \emph{channel} session types, duality is operational, swapping
$\&/\oplus$ together with $?/!$~\cite{bica2010}; for \emph{class}
session types, the residual polarity swap preserves lattice
structure (\cref{prop:duality}).

\subsection{Duality Preserves Lattice Isomorphism}\label{sec:duality}
The \emph{dual} of a session type swaps branch and selection:
$\mathit{dual}(\sbranch{m_i : S_i}) = \sselect{m_i :
\mathit{dual}(S_i)}$ and vice versa; all other constructors
($\End$, $\spar$, $\srec{X}{\cdot}$) are covariant.
Duality is an involution: $\mathit{dual}(\mathit{dual}(S)) = S$.

\begin{proposition}[Duality isomorphism]\label{prop:duality}
For every well-formed session type $S$, there is a
bounded-lattice isomorphism
\[
\statespace{S}/{\equiv} \;\cong\;
\statespace{\mathit{dual}(S)}/{\equiv}.
\]
\end{proposition}

\begin{proof}
By structural induction on $S$. The base case ($\End$) is
trivial since duality is the identity. For choice, branch and
selection produce identical state spaces (\cref{lem:polarity}),
so polarity reversal is invisible at the lattice level. For
parallel, recursion, and named definitions, duality distributes
through the constructors, and the inductive isomorphisms on
sub-types compose into one on the whole.
\end{proof}

\Cref{prop:duality} says that polarity is invisible at the lattice
level: $S$ and $\mathit{dual}(S)$ produce the same lattice
structure.
For class session types, only the object's protocol $S$ is
explicit; the dual $\mathit{dual}(S)$ specifies the implicit
second-party (client) protocol, and these two faces share the
same lattice structure.
Both tools verify $\statespace{S}/{\equiv} \cong
\statespace{\mathit{dual}(S)}/{\equiv}$ for all 108~benchmarks.
The Lean~4 mechanisation proves $\mathit{dual}(\mathit{dual}(S))
= S$ on the AST; the state-space isomorphism follows from the
constructor-wise argument above.

\camrev{The construction erases polarity (\cref{lem:polarity}), so the
reticulate does not distinguish $S$ from $\mathit{dual}(S)$;
\cref{prop:duality} thus reflects this erasure rather than a structural
symmetry. Recovering polarity\,---\,by labelling the cover
relation\,---\,would make duality a non-identity involution and realise the
linear-logic reading ($\&$ as meet $\sqcap$, $\oplus$ as join $\sqcup$); this
is left to future work on labelled reticulates.}

\subsection{Subtyping as Lattice Embedding}\label{sec:subtyping}
The lattice structure has a concrete consequence for subtype
checking. Gay and Hole~\cite{gay2005subtyping} define a coinductive
subtyping relation $\leq_{GH}$ on session types as the greatest
fixed point of a monotone operator on type-indexed binary
relations. In words, $S_1 \leq_{GH} S_2$ means the subtype $S_1$
can be substituted wherever a $S_2$-typed object is expected\,---\,a
branch subtype offers \emph{more} methods (covariant width); a
selection subtype emits \emph{fewer} labels (contravariant width);
parallel is congruent; recursion is closed under one-step
unfolding. We refer to~\cite{gay2005subtyping} for the formal
inference rules.

The contravariance of selection width is the \emph{dual} of
branch covariance: by \cref{prop:duality}, swapping polarity
preserves the lattice, so the selection case follows from the
branch case applied to the dual types. Since branch and selection
produce \emph{identical state spaces} (\cref{sec:statespace}),
the embedding relationship depends on which type has more
top-level choices\,---\,not on the polarity.

\begin{definition}[Lattice embedding]\label{def:embedding}
An \emph{embedding} of a bounded lattice $L_1$ into $L_2$ is an
injective map $\varphi : L_1 \to L_2$ that preserves and reflects
order: $a \leq b$ iff $\varphi(a) \leq \varphi(b)$.
Every embedding preserves meets and joins.
\end{definition}

\camrev{On the bare order this embedding is \emph{sound} for $\leq_{GH}$ but
not faithful: types differing only in terminal-method
width\,---\,e.g.\ $\sbranch{a : \End}$ and $\sbranch{a : \End,\; b :
\End}$\,---\,share one lattice, because \emph{end}-identification sends
both methods to the common $q_\bot$. A faithful invariant needs the labelled
transitions~$\delta$: subtyping is a label-preserving simulation, of which
the order embedding is the SCC-quotient projection. The labelled-lattice /
bisimulation correspondence is developed separately.}

\begin{proposition}[Width--embedding correspondence]%
\label{prop:embedding}
For all well-formed non-recursive session types related by
\emph{width} subtyping (shared-label continuations coincide,
$S_i = S_i'$ for $i \le k$):
\begin{enumerate}
  \item \textbf{Branch width.} If
    $\sbranch{m_1{:}S_1, \ldots, m_n{:}S_n} \leq_{GH}
    \sbranch{m_1{:}S_1', \ldots, m_k{:}S_k'}$ with $k \leq n$,
    then the supertype embeds into the subtype:
    \[
      \statespace{\sbranch{m_1{:}S_1', \ldots, m_k{:}S_k'}}/{\equiv}
      \;\hookrightarrow\;
      \statespace{\sbranch{m_1{:}S_1, \ldots, m_n{:}S_n}}/{\equiv}.
    \]
  \item \textbf{Selection width.} If
    $\sselect{l_1{:}S_1, \ldots, l_n{:}S_n} \leq_{GH}
    \sselect{l_1{:}S_1', \ldots, l_k{:}S_k'}$ with $n \leq k$,
    then the subtype embeds into the supertype:
    \[
      \statespace{\sselect{l_1{:}S_1, \ldots, l_n{:}S_n}}/{\equiv}
      \;\hookrightarrow\;
      \statespace{\sselect{l_1{:}S_1', \ldots, l_k{:}S_k'}}/{\equiv}.
    \]
\end{enumerate}
\end{proposition}

\begin{proof}
Without recursion the state space is acyclic, so the SCC quotient
is the identity and no states merge. For width subtyping the shared
continuations coincide, so the supertype's reticulate is exactly the
subtype's restricted to the shared labels; the extra labels contribute
fresh sub-lattices disjoint above $q_\bot$ (\cref{lem:choice}). This
inclusion is injective, order-preserving and order-reflecting.
Selection-width follows by duality (\cref{prop:duality}).
\end{proof}

The two clauses share one principle: the side with fewer
top-level choices has the smaller lattice and embeds into the
larger. The flip in direction\,---\,supertype-into-subtype for
$\&$, the reverse for $\oplus$\,---\,is not a property of the
lattice (which by \cref{lem:polarity} ignores polarity) but of
$\leq_{GH}$, under which a branch subtype offers more labels and a
selection subtype fewer.

\Cref{prop:embedding} isolates \emph{width} subtyping; the full
relation $\leq_{GH}$ also allows \emph{depth} subtyping
($S_i \leq_{GH} S_i'$ on shared labels). Each such continuation
contributes its own embedding, so $\leq_{GH}$ corresponds to a
\emph{family} of embeddings rather than a single map; a complete
lattice-theoretic characterisation, including recursion, is
ongoing work. Both tools implement the embedding check, and
\cref{sec:casestudy} demonstrates it on a concrete protocol.


\section{Case Study: A Concurrent File Protocol}\label{sec:casestudy}

We demonstrate \cref{thm:reticulate}, \cref{prop:duality}, and
\cref{prop:embedding} on a sequential file reader, then recover
its writer counterpart and compose the two in parallel into the
canonical concurrent read/write File.

\paragraph{Construction.}\label{sec:casestudy:construction}
The FileReader protocol\,---\,a single thread reading from open
to close, with each read returning \textit{data} until
\textit{eof}\,---\,is
\[
S_{\mathit{FileReader}} = \sbranch{\mathit{open} :
  \srec{X}{\sbranch{\mathit{read} :
  \sselect{\mathit{data} : X,\; \mathit{eof} :
  \mathit{Close}}}}}, \quad
\mathit{Close} = \sbranch{\mathit{close} : \End}.
\]
The bottom-up construction matches subterms to constructor lemmas:

\begin{center}\small
\begin{tabular}{@{}ll@{}}
\toprule
Subterm & Lemma \\
\midrule
$\End$                                                     & \cref{lem:end} \\
$\mathit{Close} = \sbranch{\mathit{close} : \End}$         & \cref{lem:choice}, $n{=}1$; \cref{lem:equations} \\
$\sselect{\mathit{data} : X,\; \mathit{eof} : \mathit{Close}}$ & \cref{lem:choice}, $n{=}2$ \\
$\sbranch{\mathit{read} : \ldots}$                         & \cref{lem:choice}, $n{=}1$ \\
$\srec{X}{\ldots}$                                         & \cref{lem:recursion}, \cref{lem:absorption} \\
$\sbranch{\mathit{open} : \ldots}$                         & \cref{lem:choice}, $n{=}1$ \\
\bottomrule
\end{tabular}
\end{center}

\noindent
The recursion step makes $q_1, q_2$ mutually reachable, so
\cref{lem:absorption} collapses them into the class $[q_1, q_2]$:
the SCC captures the \emph{read/data loop}, the protocol cycling
between read and selection indefinitely before choosing
\textit{eof}.
By \cref{thm:reticulate}, the composition witnesses that
$\statespace{S_{\mathit{FileReader}}}/{\equiv}$ is the bounded lattice
$q_0 > [q_1,q_2] > q_3 > q_\bot$ shown in \cref{fig:fileobject}(b).

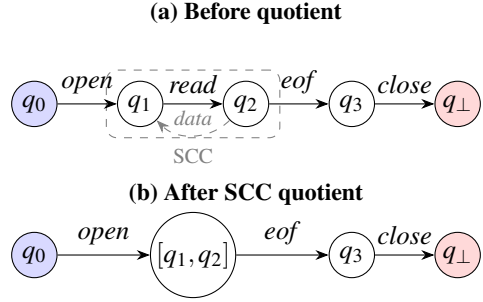
\begin{figure}[t]
\centering
\begin{tikzpicture}[every node/.style={font=\small}]
  \node[font=\footnotesize\bfseries] at (3.1,1.2) {(a) Before quotient};
  \node[state top] (q0) at (0.3,0)  {$q_0$};
  \node[state]     (q1) at (1.7,0)  {$q_1$};
  \node[state]     (q2) at (3.1,0)  {$q_2$};
  \node[state]     (q3) at (4.5,0)  {$q_3$};
  \node[state bot] (qb) at (5.9,0)  {$q_\bot$};

  \draw[trans] (q0) -- node[above] {\textit{open}}  (q1);
  \draw[trans] (q1) -- node[above] {\textit{read}}  (q2);
  \draw[trans] (q2) -- node[above] {\textit{eof}}   (q3);
  \draw[trans] (q3) -- node[above] {\textit{close}} (qb);
  \draw[trans, dashed, gray] (q2) to[bend left=45] node[above,
    font=\scriptsize, text=gray] {\textit{data}} (q1);

  \begin{scope}[on background layer]
    \node[scc, fit=(q1)(q2), label={[font=\scriptsize,text=gray]below:SCC}] {};
  \end{scope}

  \node[font=\footnotesize\bfseries] at (3.1,-1.2) {(b) After SCC quotient};
  \node[state top] (q0b) at (0.3,-2.0) {$q_0$};
  \node[state]     (scc) at (2.4,-2.0) {$[q_1,q_2]$};
  \node[state]     (q3b) at (4.5,-2.0) {$q_3$};
  \node[state bot] (qbb) at (5.9,-2.0) {$q_\bot$};

  \draw[trans] (q0b) -- node[above] {\textit{open}}  (scc);
  \draw[trans] (scc) -- node[above] {\textit{eof}}   (q3b);
  \draw[trans] (q3b) -- node[above] {\textit{close}} (qbb);
\end{tikzpicture}
\caption{%
  $\statespace{S_{\mathit{FileReader}}}$:
  (a)~before quotient, with the recursion edge $q_2 \to q_1$ dashed;
  (b)~after the SCC quotient $[q_1,q_2]$, a 4-element chain.}
\label{fig:fileobject}
\end{figure}

\paragraph{Duality.}\label{sec:casestudy:duality}
The dual swaps every choice polarity:
\[
\mathit{dual}(S_{\mathit{FileReader}}) = \sselect{\mathit{open} :
  \srec{X}{\sselect{\mathit{read} :
  \sbranch{\mathit{data} : X,\; \mathit{eof} :
  \sselect{\mathit{close} : \End}}}}}.
\]
Each choice flips agency: methods that $S_{\mathit{FileReader}}$
\emph{offers} are \emph{invoked} by
$\mathit{dual}(S_{\mathit{FileReader}})$, and outcomes the file
\emph{internally selects} are \emph{received} by the dual\,---\,the
same protocol read from the opposite endpoint.
By \cref{lem:polarity}, replaying \cref{def:construction} yields
the same Hasse diagram (\cref{fig:fileobject}(b)), witnessing
\cref{prop:duality}.

\paragraph{Subtyping.}\label{sec:casestudy:subtyping}
Consider a richer file protocol that also lets a client read
metadata and close without entering the read loop:
\[
S_{\mathit{FileReader}}^{+} = \sbranch{\mathit{open} :
  \srec{X}{\sbranch{\mathit{read} :
  \sselect{\mathit{data} : X,\; \mathit{eof} :
  \sbranch{\mathit{close} : \End}}}},\;
  \mathit{stat} : \sbranch{\mathit{close} : \End}}
\]
By branch covariance (Gay--Hole width subtyping),
$S_{\mathit{FileReader}}^{+} \leq_{GH} S_{\mathit{FileReader}}$: the
subtype offers \emph{more} top-level methods.
Applying \cref{def:construction}, the outer branch widens the top
$q_0$ with a second transition $q_0 \xrightarrow{\mathit{stat}} q_4$,
where $q_4 = \sbranch{\mathit{close}:\End}$ is a fresh state reached
\emph{outside} the read loop; the shared \textit{open} arm
reproduces Steps~1--6 unchanged. After SCC quotient the subtype's
reticulate is the five-element lattice
\[
q_0 > [q_1,q_2] > q_3 > q_\bot, \qquad q_0 > q_4 > q_\bot,
\]
with top $q_0$ and the two arms meeting at $q_\bot$ ($q_3, q_4$
incomparable). The supertype's 4-chain embeds into it as the
\textit{open} arm, $q_4$ lying outside the image\,---\,a strict,
order-preserving and -reflecting embedding. This type is recursive,
so it lies outside \cref{prop:embedding}; here the recursion sits
in the shared \textit{open} continuation, so the embedding survives.

\paragraph{FileWriter.}\label{sec:casestudy:filewriter}
The symmetric writer arises by renaming
\textit{read}/\textit{data} to \textit{write}/\textit{ack}:
\[
S_{\mathit{FileWriter}} = \sbranch{\mathit{open} :
  \srec{X}{\sbranch{\mathit{write} :
  \sselect{\mathit{ack} : X,\; \mathit{eof} :
  \sbranch{\mathit{close} : \End}}}}}.
\]
The construction is invariant under method renaming, so
$\statespace{S_{\mathit{FileWriter}}} \cong
\statespace{S_{\mathit{FileReader}}}$.

\paragraph{File.}\label{sec:casestudy:file}
The canonical concurrent read/write File bundles both inside one
object:
$S_{\mathit{File}} \triangleq S_{\mathit{FileReader}} \parallel
S_{\mathit{FileWriter}}$. Two threads sharing only a filename
each instantiate one arm independently, so two opens and two
closes on the same logical file are not errors but transitions
on disjoint arms. The reticulate is the componentwise product
$\statespace{S_{\mathit{FileReader}}} \times
\statespace{S_{\mathit{FileWriter}}}$ (\cref{lem:parallel}).

\section{Implementation, Mechanisation, and Empirical Validation}
\label{sec:impl-eval}

\paragraph{Implementation.}
We have implemented the state-space construction and lattice
checking in two independently developed tools targeting distinct
deployment surfaces: a Python~3.12 CLI for standalone protocol
analysis, and a Java~21 annotation processor
that performs the check at compile time on \texttt{@Session}-annotated
code. Both are thoroughly tested and agree on all 108~benchmarks.
The pipeline is identical in both tools: a parser converts the
session-type source (Java \texttt{@Session} annotations or text
input) to an AST; a builder applies \cref{def:construction}
clause-by-clause to produce the labelled state space; the SCC
quotient is computed via Tarjan's algorithm~\cite{tarjan1972scc};
and an all-pairs check verifies $\wedge$ and $\vee$ existence on
the quotient.
The lattice check is $O(n^2 \cdot m)$ where $n = |Q|$ and
$m = |\delta|$, completing in under 100\,ms on commodity hardware
for all benchmarks (a sub-millisecond median, and under 5\,ms for
the largest benchmark with $|Q| = 41$).

\paragraph{Test generation.}
A direct application of the lattice goes beyond verification:
automated protocol-conformance test generation. Both
implementations expose a test-generation mode that emits
JUnit~5 cases in three categories: valid paths (chains from
$q_\top$ to $q_\bot$, enumerated by bounded DFS), violations
(disabled methods at each reachable state), and incomplete
prefixes (proper prefixes of valid paths). The lattice's
boundedness guarantees a finite test suite; bottom-reachability
(\cref{prop:extrema}) guarantees every valid path terminates.
Non-terminating types are accepted under the same three-mode
flag as the lattice check; modes, default, and exact wording
are pinned identical across both tools by mirrored parity test
suites.
Valid paths are then empty, but violations and incomplete
prefixes remain. Both tools agree on the generated test sets
across the 108~benchmarks.

\paragraph{Mechanisation.}
The paper's results have been mechanically verified in
Lean~4~\cite{lean4} using Mathlib~\cite{mathlib4}.
\Cref{thm:reticulate} is mechanised as a universal theorem
\texttt{reticulate\_lattice} quantifying over every well-formed
session type, with no extra hypothesis;
duality's AST involution
$\mathit{dual}(\mathit{dual}(S)) = S$ and polarity swap are
mechanised, and \cref{prop:embedding} (width embedding) and the
forward direction of \cref{prop:realisability} are backed by
dedicated modules. The full paper-faithful artefact comprises
17~modules with zero \texttt{sorry} and zero \texttt{admit}, and
uses no coinduction. The mechanisation, the Python checker, the
benchmark corpus, and a per-result map to the Lean modules are
archived as a companion artifact~\cite{artifact}.

\paragraph{Empirical validation.}
The 108-protocol benchmark suite plays two distinct roles:
\emph{(i)}~it tests that the Python and Java implementations
\emph{conform} to the theorem on real protocols (tooling
conformance); and \emph{(ii)}~it characterises the theorem's
empirical scope\,---\,every protocol in the suite is well-formed and
produces a lattice, and the 86\%/14\% distributive split (below) gives
a first empirical picture of reticulate structure across
application domains. The 108~protocols
(\cref{tab:benchmarks}) draw from \emph{literature benchmarks}
(Java Iterator, SMTP, HTTP, JDBC, OAuth~2.0, Two-Buyer, and
others from the Mungo/StMungo/Scribble literature),
\emph{industrial protocols} (MCP, A2A, TLS, Saga Orchestrator,
WebSocket, gRPC), and \emph{synthetic benchmarks} systematically
covering edge cases (deep recursion, wide branching $n \geq 10$,
nested parallel, mixed polarity).

\begin{table}[!ht]
\centering
\footnotesize
\caption{Benchmark summary by source.}
\label{tab:benchmarks}
\begin{tabular}{lrrr}
\toprule
\textbf{Source} & \textbf{Protocols} & \textbf{Distributive} & \textbf{Non-dist.} \\
\midrule
Literature & 34 & 32 & 2 \\
Industrial & 26 & 18 & 8 \\
Synthetic  & 48 & 43 & 5 \\
\midrule
\textbf{Total} & \textbf{108} & \textbf{93} & \textbf{15} \\
\bottomrule
\end{tabular}
\end{table}

All 108~protocols produce bounded lattices under SCC quotient.
We classify each lattice in the Birkhoff
hierarchy~\cite{birkhoff1967lattice}: 93 (86\%) are
\emph{distributive} (no forbidden $N_5$ or $M_3$ sublattice)
and 15 (14\%) are non-distributive\,---\,13 containing $N_5$
(pentagon) and 3 containing $M_3$ (diamond); one protocol (Hex)
contains both.
All product lattices from parallel composition are distributive,
consistent with the lattice-theoretic result that the product of
distributive lattices is distributive.
This classification is mechanised in Lean~4 with a proof
(by \texttt{decide}) that the distributive law fails on $N_5$.

\begin{table}[t]
\centering
\footnotesize
\caption{The 15 non-distributive benchmarks.
  13 contain $N_5$ (pentagon); 3 contain $M_3$ (diamond).}
\label{tab:nondist}
\begin{tabular}{lrll}
\toprule
Protocol & $|Q|$ & Type & Cause \\
\midrule
TLS Handshake         & 13 & $N_5$ & branch-inside-selection (HELLO\_RETRY) \\
Two-Phase Commit      &  7 & $N_5$ & branch-selection nesting \\
Ki3 Onboarding        & 32 & $N_5$ & nested branch-selection \\
Ki3 CI/CD Pipeline    & 41 & $N_5$ & 4-stage branch-selection nesting \\
Quantum Measurement   &  9 & $N_5$ & crossed branch-selection \\
SSH Handshake         & 18 & $N_5$ & branch-selection nesting \\
Alternative Splicing  &  7 & $M_3$ & three branches sharing outcomes \\
QCD Gluon Exchange    &  6 & $M_3$ & three-way symmetric reconvergence \\
Hex (2$\times$2)      & 30 & both  & game-tree reconvergence + nesting \\
\bottomrule
\multicolumn{4}{l}{\footnotesize Plus 6 further benchmarks
  (biology, physics, security, games) with analogous patterns.}
\end{tabular}
\end{table}

Two structural patterns account for all 15 non-distributive cases:
$N_5$ from \emph{depth-unequal choice} (branches of a choice
constructor terminating at different depths yield a 3-chain plus
side-branch) and $M_3$ from \emph{wide choice} ($n \geq 3$
branches reconverging at $q_\bot$ yield three
pairwise-incomparable atoms; binary choice produces at most~$M_2$,
which is distributive).
Both patterns are polarity-symmetric (\cref{prop:duality}); the
empirical clustering of $N_5$ in polarity-mixed nestings and $M_3$
in wide branchings (\cref{tab:nondist}) reflects regularities of
real protocols, not structural requirements.
Characterising precisely which session types yield distributive
reticulates remains open.

We additionally enumerate all well-formed terminating session
types up to depth~3 with label sets $\{a, b\}$,
confirming that every such type produces a lattice\,---\,further
evidence for universality.

\section{Related Work}\label{sec:related}

\paragraph{Session types.}
Session types were introduced by Honda et
al.~\cite{honda1998} for binary communication and extended to
multiparty settings~\cite{honda2008multiparty};
Vasconcelos~\cite{vasconcelos2012sessions},
H{\"u}ttel et al.~\cite{huttel2016foundations}, and
Ancona et al.~\cite{ancona2016behavioral} provide foundational
treatments and comprehensive surveys.
Honda et~al.'s formulation is a process calculus with explicit
$?T$/$!T$ directionality and delegation; ours is an object-protocol
specification on a single shared object, with $\spar$ over $n$
participants on that object rather than between independent
processes.
Session types for objects were developed by
Dezani-Ciancaglini et al.~\cite{dezani2006sessions}
and Gay et al.~\cite{bica2010}, with tool
support from \mungo and \stmungo~\cite{mungo2016,stmungo2017}.
\mungo uses named definitions as the primary design
mechanism; our grammar adopts the same notation. By contrast,
our $\spar$ constructor has no counterpart in typestate tools.
Gay and Hole~\cite{gay2005subtyping} established subtyping for
session types (revisited as lattice embedding in
\cref{sec:subtyping}).
Scalas and Yoshida~\cite{scalas2019less} rebuild multiparty
session types without global types or projection, replacing them
with a behavioural safety invariant on typing contexts.
Across the work above, the focus is on the type algebra;
we study the algebraic structure of the \emph{state spaces}
induced by individual session types.

\paragraph{Intersection/union session types.}
Padovani~\cite{padovani2010intersection,padovani2009mirror} and
Acay and Pfenning~\cite{acay2017intersections} develop a
\emph{type-algebra} lattice where internal and external choice are
replaced by intersection $\wedge$ and union $\vee$, yielding a
semantic subtyping lattice $\langle 2^{\mathcal{P}},\subseteq\rangle$
of closed process-set denotations. Our lattice is orthogonal: it
lives \emph{inside} a single type, on its reachable state space.
Padovani~\cite{padovani2010intersection} (p.~80) conjectures
distributivity of his process-set lattice and leaves it open;
our 86\%/14\% distributive split on 108~benchmarks
(\cref{sec:impl-eval}) is a first empirical probe of the analogous
question on reachability lattices, with a formal bridge left to
future work.
Barbanera and de'Liguoro~\cite{barbanera2015subbehaviour} similarly
develop a sub-behaviour preorder on client/server session types via
LTS compliance, but prove duals-as-minima rather than a lattice.
Strom and Yemini's 1986 typestate proposal~\cite{strom1986typestate}
required typestates of each type to form a meet-semilattice\,---\,the
closest historical ancestor. But their lattice captures
abstract-interpretation dataflow, not protocol reachability; in
particular, our $\spar$ constructor is what forces completion
to a \emph{full} lattice.

\paragraph{Lattices in programming languages.}
Lattice theory underlies abstract
interpretation~\cite{cousot1977abstract} and is the standard
algebraic vocabulary for order in PL~\cite{birkhoff1967lattice,davey2002lattices}.
We show that \emph{behavioural} types produce lattices at the
state-space level, a distinct setting.

\paragraph{Linear logic and session types.}
Caires and Pfenning~\cite{caires2010session} established a
Curry--Howard correspondence between session types and
intuitionistic linear logic; Wadler's~\cite{wadler2012propositions}
classical version (``propositions as sessions'') links negation to
session-type duality, complementing our state-space duality
(\cref{prop:duality}). Toninho, Caires, and
Pfenning~\cite{toninho2013higher} extended the correspondence to
higher-order sessions; the lattice-theoretic perspective on
higher-order types remains open.

\paragraph{Subtyping.}
Gay and Hole~\cite{gay2005subtyping} defined subtyping for session
types; Dardha et al.~\cite{dardha2017revisited} revisited it in a
pi-calculus setting.
Padovani~\cite{padovani2016fair} extended subtyping to multiparty
types with fairness; his coinductive characterisation operates on
infinite trees, complementary to our finite quotient embedding.
Castagna and Frisch~\cite{castagna2005gentle} developed semantic
subtyping using set-theoretic models.
Our characterisation of subtyping as lattice embedding
(\cref{sec:subtyping}) provides an algebraic alternative to the
coinductive, fair, and set-theoretic approaches.

\paragraph{Multiparty session types.}
Honda, Yoshida, and Carbone~\cite{honda2008multiparty} introduced
multiparty session types (MPST), where a global type is projected
onto local types for each role.
Our Reticulate Theorem applies to local types obtained by
projection; whether projection itself is a lattice morphism
between global and local reticulates is future work.
Dagnino, Giannini, and Dezani-Ciancaglini~\cite{dagnino2023deconfined}
recently developed deconfined global types in a coinductive,
process-calculus setting proving progress and deadlock-freedom;
our setting is complementary\,---\,single object, finite-state
($\mu$-unfolded), specification-based\,---\,proving lattice structure
rather than progress.

\paragraph{Process algebra.}
Our $\|$ operates on a \emph{single shared object} rather than
independent processes~\cite{milner1989communication,sangiorgi2001pi},
requiring product lattice construction.
Kobayashi~\cite{kobayashi2006new} and Dardha and
Gay~\cite{dardha2018linear} use type systems for
deadlock-freedom; our lattices complement this by making all
interleavings explicit, guaranteeing that every concurrent path
reaches~$\bot$.

\section{Conclusion}\label{sec:conclusion}

We have proved that the state space of every well-formed session
type, quotiented by strongly connected components, forms a bounded
lattice (\cref{thm:reticulate}). The proof is constructive
and modular, with one lemma per constructor; the grammar is
parsimonious (six type constructs, three $n$-ary), and
well-formedness reduces to four clauses (\cref{def:wellformed}),
three standard in the session-type literature and one
(termination) specific to our setting.
We validated the theorem on 108~benchmark protocols using two
independently developed, thoroughly tested tools, and mechanically
verified the core proof in Lean~4 (17~modules, zero \texttt{sorry}).

This lattice structure\,---\,the \emph{reticulate}\,---\,opens several
directions.
The \emph{distributivity classification} (86\% distributive, 14\%
non-distributive) suggests finer algebraic structure to exploit,
such as canonical forms for the distributive fragment.
\camrev{\Cref{prop:realisability} answers a fragment of the converse, but the
reachability order does not always yield a lattice: where sharing reconverges
it descends to a bounded partial lattice (\cref{sec:reticulate-theorem}). Ordered
structure on session types has been observed on the subtyping
preorder~\cite{gay2005subtyping,barbanera2015subbehaviour} and on typestate
dataflow~\cite{strom1986typestate}, but a \emph{syntactic characterisation} of
which structure a type's reachability order carries\,---\,the natural
generalisation of the Reticulate Theorem\,---\,remains open.}
\camrev{More broadly, the construction reads as one direction of an
\emph{abstract interpretation}: \cref{thm:reticulate} makes it sound,
but the SCC quotient is label-blind, so the bare reticulate is not
faithful (\cref{sec:subtyping}). This suggests a graded hierarchy of
increasingly faithful reticulates, refining the abstraction until it
becomes reversible. The first step is concrete: we conjecture that replacing the SCC
quotient by a \emph{coinductive} quotient\,---\,unfolding each recursion
once and identifying states up to labelled bisimulation\,---\,would yield a
\emph{refined reticulate} that is again a bounded lattice and is
\emph{faithful}, retaining the labels, polarity, and recursive
continuity that end-identification and the SCC quotient forget, and
would characterise Gay--Hole subtyping exactly\,---\,%
$S \leq_{GH} T$ iff the width-preserving embedding between the refined
reticulates of $S$ and $T$ holds (supertype into subtype on branch,
subtype into supertype on selection).}

\camrev{A Galois connection between an implementation's concrete state
space and such a faithful reticulate would then yield a sound abstract
interpretation for protocol conformance~\cite{cousot1977abstract}.}
Our duality is structural\,---\,it identifies session types whose
state spaces coincide via polarity erasure
(\cref{lem:polarity}). Richer semantic dualities, between method
pairs (e.g., reader/writer, producer/consumer) or with directional
payload tags, fit naturally since the construction is invariant
under method renaming; \cite{bica2010}'s channel duality, viewed
through this lens, configures the two endpoints as a 2-arm
parallel composition with mutually dual arms (modulo $?/!$),
and a formal translation including the $?/!$ half is open.
Three grammar extensions preserve lattice structure and are the
subject of companion work: first-class sequential composition,
synchronous parallel, and replication.

The Reticulate Theorem suggests that session types are more
structured than previously recognised: the lattice emerges from the
syntax rather than being imposed by the theory.

\section*{Acknowledgements}

The author thanks the anonymous reviewers of ICE~2026, whose reviews
improved not only the paper but the theory itself, exposing gaps that
sharpened the results during the rebuttal round.

The author is grateful to Vasco~T.~Vasconcelos (Universidade de Lisboa)
and Uwe~Nestmann (Technische Universit\"at Berlin) for their supervision
during the author's doctoral studies: Vasconcelos for introducing the
author to session types and for the inspiration behind a
parallel-composition operator at type level, and Nestmann for the discussions on
lattice structure that shaped the reticulate. Any errors are the
author's alone.

\bibliographystyle{eptcs}
\bibliography{references}

\appendix
\section{Detailed Proofs}\label{appendix:proofs}

This appendix collects the manual proofs of the lemmas and
theorems of \cref{sec:statespace}, together with the structural
argument for \cref{prop:duality}.

\subsection*{Proof of \cref{prop:extrema}}
\begin{proof}
By induction on the construction: each case introduces
$q_\top$ and $q_\bot$ (End: single state; choice: fresh $q_\top$
plus shared $q_\bot$; recursion: back-edges preserve extrema in
the quotient; parallel: componentwise extrema).
\end{proof}

\subsection*{Proof of \cref{lem:end}}
\begin{proof}
The state space $\statespace{\End}$ consists of a single state $q$
that is both initial and terminal: $q_\top = q_\bot = q$.
The SCC quotient $\statespace{\End}/{\equiv}$ is the one-element poset $\{[q]\}$,
which satisfies all lattice axioms: $[q]$ is both top and bottom,
and $[q] \wedge [q] = [q] \vee [q] = [q]$.
\end{proof}

\subsection*{Proof of \cref{lem:absorption}}
\begin{proof}
Since $D$ is a proper upward-closed subset of $L$, $\bot \notin D$
(otherwise upward-closure would force $D = L$), so $L'$ retains
both extrema.
For $a, b \in L \setminus D$: the meet $a \wedge_L b$ lies in
$L \setminus D$ (the complement of an upward-closed set is
downward-closed), so it is preserved.
For the join: if $a \vee_L b \in L \setminus D$, it is preserved.
If $a \vee_L b \in D$, then $a \vee_{L'} b = \top$: any
non-$\top$ upper bound $u \in L \setminus D$ would satisfy
$u \geq a \vee_L b \in D$, placing $u \in D$ since $D$ is
upward-closed\,---\,contradicting $u \in L \setminus D$.
\end{proof}

\subsection*{Proof of \cref{lem:recursion}}
\begin{proof}
Adding back-edges from each placeholder $p_X$ to the initial
state $q_\top^S$ makes these states mutually reachable with
$q_\top^S$, hence they are collapsed into the same equivalence
class $[q_\top^S]$.
The set $D$ of quotient classes absorbed into $[q_\top^S]$ is
upward-closed: if $[s] \in D$ and $[s] \leq [r]$, then
$r \reach s$ (since $[r] \geq [s]$), $s \reach q_\top^S$
(since $[s] \in D$), and $q_\top^S \reach r$ (as initial
state); so $r$ and $q_\top^S$ are mutually reachable, giving
$[r] \in D$.
By the Termination clause, $[q_\bot]$ is reachable from
$[q_\top^S]$ in $\mathcal{L}(\srec{X}{S})/{\equiv}$, so
$[q_\bot] \notin D$; hence $D$ is a \emph{proper} upward-closed
subset of the body quotient.
By \cref{lem:absorption}, collapsing $D$ into a single top
element yields a bounded lattice.
\end{proof}

\subsection*{Proof of \cref{lem:equations}}
\begin{proof}
The construction connects the $k$ named sub-state-spaces by
redirecting placeholder transitions for each name $X_j$ to
$q_\top^{S_j}$.
Non-recursive references (acyclic names) are effectively inlined;
when non-reconvergent (the lemma's hypothesis) they preserve
lattice structure, reconvergent sharing being the
bounded-partial-lattice case of \cref{sec:reticulate-theorem}.
Cyclic references create back-edges that merge states into
SCCs, exactly as in the single-recursion case.
The set of quotient classes absorbed into each SCC is
upward-closed by the same argument as \cref{lem:recursion}.
Applying \cref{lem:absorption} iteratively (once per SCC created
by cyclic references) yields a bounded lattice. The SCC quotient is
determined by the reachability graph alone
(\cref{def:scc-quotient}); independence from the absorption order
is a confluence property mechanised as
\texttt{NamedDefinitions.\allowbreak absorption\_\allowbreak order\_\allowbreak independent}.
\end{proof}

\subsection*{Proof of \cref{lem:parallel}}
\begin{proof}
By the parallel case of \cref{def:construction} and the independence of
the $n$ branches (Termination ensures every branch reaches
$\End$; Parallel closedness ensures no shared variables), the
product inherits SCC structure componentwise\,---\,an SCC in the product is a product of
SCCs from each component:
\[
\statespace{S_1 \spar \cdots \spar S_n}/{\equiv} \;\cong\;
\textstyle\prod_{i=1}^{n}
(\statespace{S_i}/{\equiv_i}).
\]
The product of finitely many bounded lattices is a bounded
lattice (meets and joins are componentwise).
See \cref{fig:product} for an illustration with $n = 2$.
\end{proof}

\subsection*{Proof of \cref{thm:reticulate}}
\begin{proof}
By structural induction on~$S$, applying
\cref{lem:end,lem:choice,lem:recursion,lem:equations,lem:parallel}
case-by-case. Each case preserves the bounded lattice property,
completing the induction.
\end{proof}

\subsection*{Proof of \cref{prop:duality}}
\begin{proof}
By structural induction on $S$.

\emph{Base case} ($S = \End$):
$\mathit{dual}(\End) = \End$, so the state spaces are identical.

\emph{Choice} ($S = \sbranch{m_1{:}S_1, \ldots, m_n{:}S_n}$):
$\mathit{dual}(S) = \sselect{m_1{:}\mathit{dual}(S_1), \ldots,
m_n{:}\mathit{dual}(S_n)}$.
Branch and selection produce identical state spaces
(\cref{lem:polarity}), so it suffices that the continuations agree:
by induction $\statespace{S_i}/{\equiv} \cong
\statespace{\mathit{dual}(S_i)}/{\equiv}$ for each $i$, and the $n$
sub-isomorphisms with the identity on $q_\top, q_\bot$ give the
isomorphism of the whole choice state space.

\emph{Recursion} ($S = \srec{X}{S'}$):
$\mathit{dual}(\srec{X}{S'}) = \srec{X}{\mathit{dual}(S')}$.
The recursion construction redirects placeholder transitions
to the body's initial state.
Since dual acts on the body's constructors (swapping $\&$ and
$\oplus$) but preserves the placeholder and entry-state
structure, the same back-edges are created in both
$\statespace{\srec{X}{S'}}$ and
$\statespace{\srec{X}{\mathit{dual}(S')}}$, yielding
isomorphic SCC quotients by induction on the body.

\emph{Parallel} ($S = S_1 \spar \cdots \spar S_n$):
dual distributes over parallel, so
$\mathit{dual}(S) = \mathit{dual}(S_1) \spar \cdots \spar
\mathit{dual}(S_n)$.
By induction each factor is isomorphic; the product of
isomorphic lattices is isomorphic.

\emph{Named definitions}
($D = S, \; X_1 = S_1, \ldots, X_k = S_k$):
as in the recursion case, $\mathit{dual}$ acts on each equation body
but preserves the name-reference structure, so the same redirections
occur in both $\statespace{D}$ and $\statespace{\mathit{dual}(D)}$;
the SCC quotients are isomorphic by induction on each $S_i$.
\end{proof}


\end{document}